\documentclass[a4paper,11pt]{amsart}

\usepackage{amssymb}
\usepackage{mhequ}
\usepackage[latin1]{inputenc}
\renewcommand{\div}{\mbox{div}}

\usepackage{cite}

\DeclareFontFamily{OT1}{rsfs}{}
\DeclareFontShape{OT1}{rsfs}{m}{n}{ <-7> rsfs5 <7-10> rsfs7 <10-> rsfs10}{}
\DeclareMathAlphabet{\mathscr}{OT1}{rsfs}{m}{n}

\newcommand{\eq}[1]{\eqref{#1}}

\newcommand{\bel}[1]{\begin{equation}\label{#1}}
\newcommand{\beal}[1]{\begin{eqnarray}\label{#1}}
\newcommand{\beadl}[1]{\begin{deqarr}\label{#1}}
\newcommand{\eeadl}[1]{\arrlabel{#1}\end{deqarr}}
\newcommand{\eeal}[1]{\label{#1}\end{eqnarray}}
\newcommand{\eead}[1]{\end{deqarr}}
\newcommand{\eea}{\end{eqnarray}}
\newcommand{\eeaa}{\end{eqnarray*}}

\newcommand{\be}{\begin{equation}}
\newcommand{\ee}{\end{equation}}

\DeclareFontFamily{OT1}{rsfs}{}
\DeclareFontShape{OT1}{rsfs}{m}{n}{ <-7> rsfs5 <7-10> rsfs7 <10->
rsfs10}{} \DeclareMathAlphabet{\mycal}{OT1}{rsfs}{m}{n}

\newcounter{mnotecount}[section]

\newcommand{\N}{{\Bbb N}}

\newcommand{\rmnote}[1]{}

\newcommand{\Ric}{\operatorname{Ric}}

\def\mysavedown#1{\edef\mysubs{\mysubs#1}}
\def\mysaveup#1{\edef\mysups{\mysups#1}}
\def\mydown#1{{\mytensor}_{\vphantom{\mysubs}#1}}
\def\myup#1{{\mytensor}^{\vphantom{\mysups}#1}}
\def\tensor#1#2{
  #1
  \def\mytensor{\vphantom{#1}}
  \def\mysubs{\relax}
  \def\mysups{\relax}
  \let\down=\mysavedown
  \let\up=\mysaveup
  #2
  \let\down=\mydown
  \let\up=\myup
  #2
  }

\newcommand{\Riem}{\operatorname{Riem}}

\newcommand{\Tr}{\operatorname{Tr}}

\renewcommand{\div}{\operatorname{div}}

\renewcommand{\phi}{\varphi}
\renewcommand{\epsilon}{\varepsilon}

\def\crn#1#2{{\vcenter{\vbox{
        \hbox{\kern#2pt \vrule width.#2pt height#1pt
           }
          \hrule height.#2pt}}}}

\renewcommand{\hbar}{{\overline h}}

\newcommand{\pre}[2]{{{\vphantom{#2}}^{#1}}\kern-.2ex{#2}}

\theoremstyle{plain}
\newtheorem{theorem}{Theorem}[section]

\newtheorem{proposition}[theorem]{Proposition}

\theoremstyle{definition}
\newtheorem*{remark}{Remark}

\numberwithin{equation}{section}

\date{March 17, 2008}

\begin{document}
\title[Extended constraint equations on AH manifolds]
{Perturbative solutions to the extended constant scalar
curvature equations on asymptotically hyperbolic manifolds}

\author[E. Delay]{Erwann Delay}
\address{Erwann Delay,
Institut de Math\'ematiques et Mod\'elisation de Montpellier,
UMR 5149 CNRS, Universit\'e Montpellier II,
Place Eug\`ene Bataillon, 34095 Montpellier cedex 5, France}
\email{Erwann.Delay@univ-avignon.fr}
\urladdr{http://www.math.univ-avignon.fr/Delay}

\begin{abstract}
We prove local existence  of  solutions to the
extended constant scalar curvature equations introduced by A.~Butscher, in the asymptotically hyperbolic setting. This gives a new local construction
of
asymptotically hyperbolic metrics with constant scalar
curvature.
\end{abstract}

\maketitle

\\
\newline

\section{Introduction}\label{section:intro}
The study of constant scalar curvature metrics plays a
particulary important part in riemannian geometry and in
general relativity. For instance in the riemannian point of vue it
gives rise to the well known Yamabe problem when working in a
conformal class. For the general relativity, constant scalar
curvature riemannian metrics are particular solutions  to the
constraint equations, in the time symmetric case. In
\cite{Butscher:extended}, A. Butscher introduce a system of
equations called the "extended constraint equations", equivalent to
the usual constraints. This new system has the advantage that, as
the Einstein equation for riemannian metrics, it can naturally be
modified to a determined elliptic system by adding a
 gauge-breaking term.

In the asymptotically flat setting, A. Butscher proves in \cite{Butscher:extended}
that his system can be solved near the usual euclidian metric.

The present note is a study of the extended constraint equation on
asymptotically hyperbolic (A.H.) manifolds, in the time symmetric
case, near an Einstein metric. It will appear that
the proof is more simple in our context due to the fact we do not
have to handle with a kernel. Let us now introduce some
notations and the system we want to study. For a riemannian metric
$g$ on a manifold $M^n$, let us denote by $R(g)$ the scalar
curvature of $g$. In the asymptotically hyperbolic setting, the
constant scalar curvature equation is \bel{Rcte} R(g)=-n(n-1). \ee
The extended scalar curvature equation is \bel{extend1}
\Ric(g)+(n-1)g=S, \ee \bel{extend2} \div_gS=0, \ee where, for a
metric $g$, a symmetric two tensor $T$ and a one form $\xi$, \bel{S}
S=T-\frac{1}{n}\Tr_g(T)g+\mathring{\mathcal L}\xi, \ee is a trace free
symmetric two tensor, and $\mathring{\mathcal L}$ is the conformal
killing operator (see section \ref{sec:def}). At this
stage it is important to remark that equation \eq{extend2} is a
consequence of \eq{extend1} by the Bianchi identity. Taking the
trace of equation \eq{extend1} we see that any solution $g$ has
constant scalar curvature satisfying \eq{Rcte}. Reciprocally  any
metric of constant scalar curvature $R(g)=-n(n-1)$ is a solution of
\eq{extend1} with, for instance, $\xi=0$ and
$T=\Ric(g)+(n-1)g$\footnote{Of course, we can also take $T=\Ric(g)$
but our choice here is more natural in the A.H. context}.

We fix  a smooth asymptotically hyperbolic metric $g_0$.
The weighted H\"older spaces $C^{k,\alpha}_s$ we will work with are the one introduced by Lee \cite{Lee:fredholm}. It consist of tensor field of the form $\rho^s u$ for
$u$ in the usual H\"older space $C^{k,\alpha}$.

The
theorem we will prove is the following
\begin{theorem}\label{maintheorem}
Let $s\in(0,n-1)$, $k\in\N\backslash\{0\}$, $\alpha\in(0,1)$. Let
$g_0$ be a non-degenerate A.H. Einstein metric. For
all sufficiently small $T\in C^{k,\alpha}_s$,  there exists  $(\xi,h)$ close to zero in
$C^{k+1,\alpha}_s\times C^{k+2,\alpha}_s$ such that $g=g_0+h$ and
$\xi$ solves \eq{extend1} and \eq{extend2} (then $g$
solves \eq{Rcte}). The map
$$
\begin{array}{llll}
{\mathcal C}: &C^{k,\alpha}_s&\longrightarrow &C^{k+1,\alpha}_s\times C^{k+2,\alpha}_s\\
&T&\mapsto&(\xi,h)
\end{array}
$$
 is smooth near zero.
\end{theorem}
Here we say that the metric $g$ is {\it non-degenerate} if the
$L^2$-kernel of $\Delta_L+2(n-1)$ acting on trace free symmetric two tensors is trivial,
$\Delta_L$ being the Lichnerowicz laplacian of $g$ (see section \ref{sec:def}). This condition is satisfied for instance on the hyperbolic space but
also on a large class of A.H. manifolds (see \cite{Lee:fredholm} and Theorem 2.2 of \cite{ACD2}).

With the construction we use, the metric $g$ we obtain also satisfies that the identity
map from $(M,g)$ to $(M,g_0)$ is harmonic.

\medskip

{\small\sc Acknowledgments}. {\small }
I am grateful to F. Gautero for his comments on the original manuscript.

\section{Definitions, notations and conventions}\label{sec:def}

Let $(\overline{M},g)$ be a smooth, compact $n$-dimensional manifold
with boundary $\partial_\infty M$.
Let $M:=\overline{M}\backslash\partial_\infty{M}$ which is a non-compact
manifold.
We call $\partial_\infty M$
the {\it boundary at infinity} of $M$.
 Let $g$ be
a Riemannian metric on $M$. We say that $(M,g)$ is {\it
conformally compact} if there exists a smooth defining function
$\rho$ on $\overline{M}$ (that is $\rho\in C^\infty(\overline{M})$,
$\rho>0$ on $M$, $\rho=0$ on $\partial_\infty {M}$ and $d\rho$
is nowhere vanishing on $\partial_\infty M$) such that
$\overline{g}=\rho^{2}g$ is a smooth Riemannian metric on
$\overline{M}$. If $|d\rho|_{\overline{g}}=1$ on
$\partial_\infty M$, it is well known that $g$ has asymptotically
sectional curvature $-1$ (see \cite{Mazzeo:hodge}  for example) near
its boundary at infinity. In this case we say that $(M,g)$ is
{\it asymptotically hyperbolic}.

If we moreover assume that $g$ has constant scalar
curvature then asymptotic hyperbolicity enforces the
normalisation
$$
R(g)=-n(n-1),
$$
where $R(g)$ is the scalar curvature of $g$.
Also if $g$ is Einstein then the Ricci curvature of $g$ is
$$
\Ric(g)=-(n-1)g.
$$
We denote by $\nabla$ the  Levi-Civita connexion of $g$ and by
$\Riem(g)$ the Riemannian sectional curvature of $g$.

We denote by ${\mathcal T}_p$ the set of rank $p$ covariant tensors.
When $p=2$, we denote by ${\mathcal S}_2$ the subset of symmetric
tensor which splits as ${\mathcal G}\oplus {\mathring{\mathcal
S}_2}$ where ${\mathcal G}$ is the set of $g$-conformal tensors and
${\mathring{\mathcal S}_2}$ is the set of trace-free
tensors (relatively to $g$). We observe the summation
convention (the corresponding indices run from $1$ to $n$), and we
use $g_{ij}$ and its inverse $g^{ij}$ to lower or raise indices.

The Laplacian is defined as
$$
\triangle=-tr\nabla^2=\nabla^*\nabla,
$$
where $\nabla^*$ is the $L^2$ formal adjoint of $\nabla$. The
Lichnerowicz Laplacian acting on symmetric covariant 2-tensors is
$$
\triangle_L=\triangle+2(\Ric-\Riem),
$$
where
$$(\Ric\; u)_{ij}=\frac{1}{2}[\Ric(g)_{ik}u^k_j+\Ric(g)_{jk}u^k_i],
$$
and
$$
(\Riem \; u)_{ij}=\Riem(g)_{ikjl}u^{kl}.
$$
For $u$ a covariant 2-tensorfield on $M$ we define the divergence of
$u$ by $$ (\mbox{div}u)_i=-\nabla^ju_{ji}.$$ For a one-form
$\omega$ on $M$, we define the divergence of $\omega$ :
$$
d^*\omega=-\nabla^i\omega_i,
$$
the symmetric part of its covariant derivative :
$$
({\mathcal
L}\omega)_{ij}=\frac{1}{2}(\nabla_i\omega_j+\nabla_j\omega_i),$$
(note that ${\mathcal L}^*=\mbox{div}$) and the trace free part of
that last tensor :
$$
(\mathring{\mathcal
L}\omega)_{ij}=\frac{1}{2}(\nabla_i\omega_j+\nabla_j\omega_i)+\frac{1}{n}(d^*\omega)
g_{ij}.$$ The operator $\mathring{\mathcal L}$ is sometimes called
the conformal killing operator. Its formal $L^2$ adjoint
acting on trace free symmetric two tensors is
$\mathring{\mathcal L}^*=\div$.

\section{Gauge-broken equation and proof}
It is well known that the system \eq{extend1},\eq{extend2} is not
elliptic because of the invariance by diffeomorphism. As usual
in this context, we will add a gauge term to the system in such a
way that it becomes elliptic, and that the solutions to the
new system are solutions to the original one.

Let us define an operator from symmetric two tensors to one-forms :
$$
B_g(h)=\div_gh+\frac{1}{2}d(\Tr_gh).
$$

The new system we consider is
\bel{extend1b}
\Ric(g)+(n-1)g-{\mathcal L}_g(B_g(g_0))=S,
\ee
\bel{extend2b}
\div_gS=0,
\ee
where $S$ is as in \eq{S}.
First let us verify that the solutions to the new system are solution
to the original one.
\begin{proposition}\label{sol}
Let $s\in[0,n)$, $k\in\N\backslash\{0\}$, $\alpha\in(0,1)$.
Let $g_0$ be an A.H. metric with negative Ricci curvature.
If  $h\in C^{k+2,\alpha}_s$ is sufficiently small, and the metric
$g=g_0+h$ is a solution of \eq{extend1b}\eq{extend2b}, then it is a solution
of  \eq{extend1}\eq{extend2}.
\end{proposition}
\begin{proof}
We will apply the operator $B_g$ to the equation \eq{extend1b}. We
remark that $\div S=\Tr S=0$, that $B_g(\Ric(g))=0$ by Bianchi
identity, and that $B_g(g)=0$. Then, if we denote by $\omega$ the
one-form $B_g(g_0)$, we obtain that
$$
B_g({\mathcal L}_g\omega)=0.
$$
This equation reads in local coordinate:
$$
-\nabla^{i}\left[\frac{1}{2}(\nabla_i\omega_j+\nabla_j\omega_i)\right]
+\frac{1}{2}\nabla_j\nabla^i\omega_i=0.
$$
Commuting derivatives and multiplying by $2$, we obtain that
\bel{omega} \Delta_g\omega-\Ric(g)\omega=0. \ee As $\Ric(g_0)$ is
negative, the operator $\Delta_{g_0}-\Ric(g_0)$ has no $L^2$ kernel
on one-forms so from \cite{Lee:fredholm}, it is an isomorphism from
$C^{k+1,\alpha}_t$ to $C^{k-1,\alpha}_t$, for all $t\in(-1,n)$. Now
if $h$ is small in $C^{k+2,\alpha}_s$, then  it is small in
$C^{k+2,\alpha}_0$, so the operator $\Delta_{g}-\Ric(g)$ is
still an isomorphism between the same spaces. As $\omega\in
C^{k+1,\alpha}_s$ with $s\in[0,n)$, we conclude that $\omega=0$.
\end{proof}

\begin{remark} 
The fact that $B_g(g_0)$ vanishes show that the identity map from $(M,g)$ to
$(M,g_0)$ is harmonic (see \cite{GL} for instance).
\end{remark}

Let us now construct  solutions  to the system \eq{extend1b}\eq{extend2b} by
an implicit function theorem on Banach spaces.
\begin{proposition}\label{gaugesol}
Let $s\in(0,n-1)$, $k\in\N\backslash\{0\}$, $\alpha\in(0,1)$.
Let $g_0$ be a non degenerate A.H. Einstein metric.
Then for all $T\in C^{k,\alpha}_s$ sufficiently small, there exist a unique
$(\xi,h)$ close to zero in  $C^{k+1,\alpha}_s\times C^{k+2,\alpha}_s$ such that $g=g_0+h$ and
$\xi$ solve \eq{extend1b} and \eq{extend2b}.
\end{proposition}

\begin{proof}
Let us consider the map from a neighborhood of zero
in
$C^{k+2,\alpha}_s\times C^{k+1,\alpha}_s \times C^{k,\alpha}_s$ to
$C^{k,\alpha}_s\times C^{k-1,\alpha}_s$ defined by
$$
F(h,\xi,T):=\left(
\begin{array}{c}\Ric(g)+(n-1)g-{\mathcal L}_g(B_g(g_0))-S\\
\div_g S\end{array}\right),
$$
where $g=g_0+h$  and $S$ is defined
by \eq{S}. The map $F$ is well defined and differentiable in a
neighborhood of zero. We also have $F(0,0,0)=0$ and the derivative
of $F$ in the first two variables at the origin  is (this is at this stage
we use $g_0$ is Einstein, see \cite{Delay:study} for instance)
$$
D_{(h,\xi)}F(0,0,0)(\delta h,\delta\xi)=\left(
\begin{array}{c}\frac{1}{2}\Delta_L(\delta h)+(n-1)(\delta h)-\mathring{\mathcal L}(\delta\xi)\\
\div\circ\mathring{\mathcal L} (\delta\xi)\end{array}\right),
$$
where all the operators are relative to the metric $g_0$. The
hypothesis that $g_0$ is non degenerate together with the results of
\cite{Lee:fredholm} give us  that the operator
$\frac{1}{2}\Delta_L+(n-1)$ is an isomorphism from
$C^{k+2,\alpha}_s$ to $C^{k,\alpha}_s$ when $s\in(0,n-1)$. Now  the
operator $\mathring{\mathcal L}$ has no $L^2$ kernel
(\cite{AndChDiss}, \cite{Gicquaud1}) so the same is true for the operator
$\div\circ\mathring{\mathcal L}$ because $\div=\mathring{\mathcal
L}^*$. Again the result of \cite{Lee:fredholm} shows that
$\div\circ\mathring{\mathcal L}$ is an
isomorphism from
$C^{k+1,\alpha}_\sigma$ to $C^{k-1,\alpha}_\sigma$ when
$\sigma\in(-1,n)$. We then conclude that  $D_{(h,\xi)}F(0,0,0)$ is an
isomorphism and the proposition follows by the implicit function
theorem.
\end{proof}
Propositions \ref{gaugesol} and \ref{sol} together prove the main theorem
\ref{maintheorem}

\bibliographystyle{amsplain}

\bibliography{../references/newbiblio,%
../references/reffile,%
../references/bibl,%
../references/hip_bib,%
../references/newbib,%
../references/PDE,%
../references/netbiblio,%
../references/erwbiblio,%
 stationary}

\end{document}